\documentclass[11pt]{article}

\title{\bf{Walsh Sampling with Incomplete Noisy Signals}}

\author{\textsc{Yi~Janet Lu}\\
National Research Center of Fundamental Software, Beijing, P.R.China\\
Department of Informatics, University of Bergen, Bergen, Norway\\
Email: {yi.janet.lu@gmail.com}
}

\usepackage[noadjust]{cite}
\usepackage[usenames,dvipsnames]{xcolor}
\usepackage{fullpage}

\usepackage{amsmath,amssymb}
\usepackage{amsthm}

\usepackage{url}
\usepackage{epsfig,graphics,color,subfig}

\newtheorem{theorem}{Theorem}
\newtheorem{corollary}{Corollary}

\newtheorem{proposition}{Proposition}
\newtheorem{conjecture}{Conjecture}

\usepackage{listings}

\begin{document}

\lstdefinestyle{custom}{
 captionpos=t,
 belowcaptionskip=1\baselineskip,
 breaklines=true,
 frame=L,
 xleftmargin=\parindent,
 showstringspaces=false,
 basicstyle=\footnotesize\ttfamily,
 keywordstyle=\bfseries\color{green!40!black},
 commentstyle=\itshape\color{purple!40!black},
 identifierstyle=\color{black},
 stringstyle=\color{orange}
}

\lstset{style=custom}

\maketitle

\begin{abstract}
With the advent of massive data outputs at a regular rate, admittedly, signal processing technology plays an increasingly key role. Nowadays, signals are not merely restricted to physical sources, they have been extended to digital sources as well. 

Under the general assumption of discrete statistical signal sources, we propose a practical problem of sampling incomplete noisy signals for which we do not know a priori and the sampling size is bounded. We approach this sampling problem by Shannon's channel coding theorem. Our main results demonstrate that it is the large Walsh coefficient(s) that characterize(s) discrete statistical signals, regardless of the signal sources. By the connection of Shannon's theorem, we establish the necessary and sufficient condition for our generic sampling problem for the first time. Our generic sampling results find practical and powerful applications in not only statistical cryptanalysis, but software system performance optimization.

\end{abstract}

\section{Introduction}

With the advent of massive data outputs regularly,
we are confronted by the challenge of
big data processing and analysis. Admittedly,
signal processing has become an increasingly key technology. 
An open question is the sampling problem with the signals, for which we assume that we do not know \emph{a priori}.
Due to reasons of practical consideration, sampling is affected
by possibly strong noise and/or the limited measurement precision.
Assuming that the signal source is not restricted to a particular application domain,
we are concerned with a practical and generic problem to sample these noisy signals.

Our motivation arises from the following problem in modern applied statistics.
Assume the discrete statistical signals in a general setting as follows.
The samples, generated by an arbitrary (possibly noise-corrupted) source $F$,
 are $2^n$-valued for a fixed $n$. 
It is known to be a hypothesis testing problem
 to test presence of any signals.  
Traditionally, $F$ is a deterministic function with small or medium input size.
It is computationally easy to collect the \emph{complete and precise} distribution $f$ of $F$.
Based on the notion of Kullback-Leibler distance,
 the conventional approach (aka. the classic distinguisher) 
solves the sampling problem, given the distribution $f$ \emph{a priori} (see \cite{vaudenay_textbook2006}).
Nevertheless, in reality, $F$ might be
 a function that we do not have the complete description, or it might have large input size,
or it maybe a non-deterministic function.
Thus, it is infeasible to collect the complete and precise distribution $f$.
This gives rise to the new generic statistical
sampling problem with discrete incomplete noisy signals, using the bounded number of samples.

In this work, we show that we can solve the generic sampling problem as reliable as possible 
without knowing signals \emph{a priori}.
By novel translations, Shannon's channel coding theorem 
can solve the generic sampling problem 
under the general assumption of statistical signal sources.  
Specifically, the \emph{necessary and sufficient} condition is given \emph{for the first time} 
to sample the incomplete noisy signals with bounded sampling size for signal detection.
It is interesting to observe that the classical signal processing tool of Walsh transform
\cite{dsp_book,walsh-book2} is essential: 
regardless of the signal sources, it is the large Walsh coefficient(s) that
 characterize(s) discrete statistical signals.
Put other way, when sampling incomplete noisy signals of the same source multiple times,
one can expect to see \emph{repeatedly} those large Walsh coefficient(s) of same magnitude(s) at the
fixed frequency position(s). Note that this is known in application
domains such as images, voices.
 Our results show strong connection between Shannon's theorem and Walsh transform, 
both of which are the key innovative technologies in digital signal processing.
Our generic sampling results find practical and useful applications in not only statistical cryptanalysis - it is expected to become a powerful universal analytical tool for the core building blocks of symmetric cryptography (cf. \cite{my_new_submission,vaudenay_new}),
but performance analysis and heterogeneous acceleration. 
The latter seems to be one of the main bottlenecks for large-scale IT systems 
 in the era of the revolutionary development of memory technologies.

The rest of the paper is organized as follows. 
In Section \ref{sect_review_statistics},
we review the basics of Walsh transforms and its application of multi-variable
tests in statistics. 
In Section \ref{sect_shannon},
Shannon's famous channel coding theorem, also known as Shannon's Second Theorem,
is reviewed.
In Section \ref{sect_main_result},
we present our main sampling results: we
put forward two sampling problems,
namely, the classical and generic versions;
we also conjecture a quantitative relation between Renyi's divergence of degree 1/2 and Shannon's channel capacity.
In Section \ref{sect_experiments}, we give illustrative applications and experimental results.
Finally, we give conclusions and future work in Section \ref{sect_end}.
\section{Walsh Transforms in Statistics}
\label{sect_review_statistics}

Given a real-valued function 
$f: GF(2)^n \to \rm{R}$, which is defined on an $n$-tuple binary vector of input,
the Walsh transform of $f$, denoted by $\widehat{f}$, is another real-valued function defined as
\begin{equation}
\label{E_walsh_def_function}
\widehat{f}(i) \stackrel{\text{def}}{=} \sum_{j\in GF(2)^n}(-1)^{\langle i,j\rangle} f(j),
\end{equation}
for all $i \in GF(2)^n$, where $<i,j>$ denotes the inner product between two $n$-tuple binary vectors $i,j$.
For later convenience, we give an alternative definition below.
Given an input array $x=(x_0,x_1,\ldots,x_{2^{n}-1})$ of $2^n$ reals in the time domain, 
the Walsh transform $y= \widehat{x} =(y_0,y_1,\ldots,y_{2^{n}-1})$ of $x$ is defined by
\begin{equation}
\label{E_walsh_def_array}
y_i \stackrel{\text{def}}{=} \sum_{j\in GF(2)^n} (-1)^{\langle i,j\rangle} x_j,
\end{equation}
for any $n$-tuple binary vector $i$. 
We call $x_i$ (resp. $y_i$) the time-domain component (resp. transform-domain coefficient) of 
the signal with dimension $2^n$.
For properties and references on Walsh transforms, 
we refer to \cite{walsh-book2,my_new_submission,Vetterli2015}.

Let $f$ be a probability distribution of an $n$-bit random variable 
$\mathcal{X}=(X_n,X_{n-1},\ldots,X_1)$, where each $X_i\in \{0,1\}$.
Then, $\widehat{f}(m)$ is the \emph{bias} of the Boolean variable $\langle m,\mathcal{X}\rangle$ for any fixed 
$n$-bit vector $m$,
which is often called the output \emph{pattern} or \emph{mask} (and note that the pattern $m$ should be nonzero). 
Here,
recall that a Boolean random variable $\mathcal{A}$ has \emph{bias} $\epsilon$, which is defined by
\begin{equation}
\epsilon \stackrel{\text{def}}{=} E[(-1)^\mathcal{A}] = \Pr(\mathcal{A}=0)-\Pr(\mathcal{A}=1).
\end{equation}
Hence, we always have $-1 \le \epsilon \le 1$ and
if $\mathcal{A}$ is uniformly distributed, $\mathcal{A}$ has bias 0.

Walsh transforms were used 
in statistics to find dependencies within a multi-variable data set.
In the multi-variable tests, each $X_i$ indicates the presence or absence (represented by `1' or `0') 
 of a particular feature in a pattern recognition experiment.
Fast Walsh Transform (FWT) is used to obtain all coefficients $\widehat{f}(m)$ in one shot.
By checking the Walsh coefficients one by one and identifying the \emph{large}\footnote{Throughout the paper,
 we refer to the large transform-domain coefficient $d$ as the one
  with a large absolute value.} ones, 
we are able to tell the dependencies among $X_i$'s. 
For instance,
we have a histogram of the pdf (the probability density function) of the triples $(X_0, X_1, X_2)$ as depicted in Fig. 1.
It is obtained by an experiment of trying a total of $160$ times.
The Walsh spectrum of the histogram, i.e., of the array $(24,18,16,26,22,14,16,24)$,
as calculated by (\ref{E_walsh_def_array}),
 is shown in Fig. 2.
The largest nontrivial Walsh coefficient is found to be $32$ located at the index position $(0,1,1)$.
This implies that the correlation is observed to be strongest for the variables $X_1, X_2$.

\begin{figure}[!t]
\begin{center}
\includegraphics[scale=0.6]{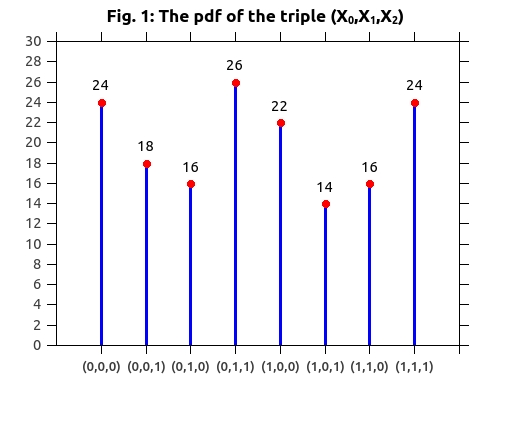}
\end{center}
\end{figure}

\begin{figure}[!t]
\begin{center}
\includegraphics[scale=0.6]{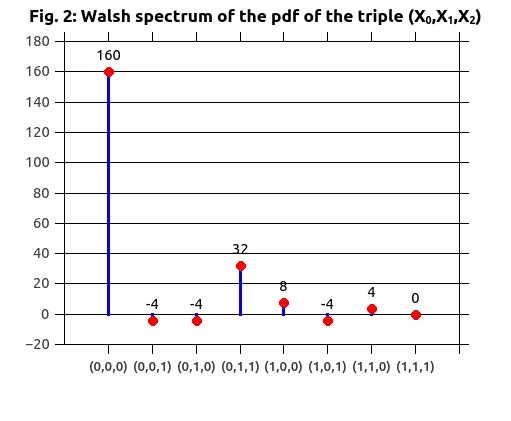}
\end{center}
\end{figure}

\section{Review on Shannon's Channel Coding Theorem}
\label{sect_shannon}

We briefly review Shannon's famous channel coding theorem\footnote{Sometimes it's called Shannon's Second Theorem.}
 (cf. \cite{it-book}).
First, we recall basic definitions of Shannon entropy.
The entropy $H(X)$ of a discrete random variable $X$ with alphabet
$\mathcal{X}$ and probability mass function $p(x)$
is defined by
\[
H(X) \stackrel{\text{def}}{=} -\sum_{x\in \mathcal{X}} p(x)\log_2 p(x).
\]
The joint entropy $H(X_1,\ldots,X_n)$ of a collection of discrete random variables $(X_1,\ldots,X_n)$
with a joint distribution $p(x_1, x_2, \ldots, x_n)$ is defined by
\[
H(X_1,\ldots,X_n) \stackrel{\text{def}}{=} - \sum_{\substack{x_1\\ \cdots\\ x_n}} p(x_1, \ldots, x_n)\log_2 p(x_1, \ldots, x_n).
\]
Define the conditional entropy $H(Y|X)$ of a random variable $Y$ given $X$ by
\[
H(Y|X) \stackrel{\text{def}}{=} \sum_x p(x)H(Y|X=x).
\]
The mutual information $I(X;Y)$ between
two random variables $X,Y$ is equal to $H(Y)-H(Y|X)$, which always equals $H(X)-H(X|Y)$.
 A communication channel is a system in which the output $Y$ depends
probabilistically on its input $X$. It is characterized by a probability
transition matrix that determines the conditional distribution of the
output given the input.

\begin{theorem}[Shannon's Channel Coding Theorem]
\label{thm_1}
Given a channel, denote the input, output by $X,Y$ respectively.
We can send information at the maximum rate $C$ 
bits per transmission with an arbitrarily low probability of error, where $C$ is the channel capacity defined by 
\begin{equation}\label{E_capacity_def}
C = \max_{p(x)} I(X;Y), 
\end{equation}
and the maximum is taken over all possible input distributions $p(x)$.
\end{theorem}

For the binary symmetric channel (BSC) with crossover probability $p$, that is, 
the input symbols are complemented
with probability $p$,
so the transition matrix is of the form
\begin{equation}
\left(
\begin{array}{cc}
1-p & p \\
p & 1-p \\
\end{array}
\right)
\,
.
\end{equation}
We can express $C$ by (cf. \cite{it-book}):
\begin{equation}\label{E_capacity}
C = 1 - H(p) 
\textrm{ bits/transmission.}
\end{equation}

We refer to the BSC with crossover probability $p = (1 + d)/2$ and $d$ is small (i.e., $|d|\ll 1$) as an 
\emph{extremal BSC}.
Note that
such a BSC becomes useful in studying the correlation attacks (see \cite{Meier1989correlation_attacks}) on stream ciphers.
To put simply, the observed keystream bit is found to be correlated with one 
bit which is a linear relation of the internal state.
In this case,
the probability that the equality holds is denoted by $1-p$ (and so the crossover probability of this BSC is $p$).
Using 
\begin{equation}
H\Bigl(\frac{1+d}{2}\Bigr) = 1 - \Bigl(\frac{d^2}{2} + \frac{d^4}{12} + \frac{d^6}{30} + \frac{d^8}{56} + 
\underbrace{\cdots}_{O(d^{10})}
 \Bigr) 
\times \frac{1}{\log 2},
\end{equation}
we can show
\begin{equation}\label{approx_H_special}
H\Bigl(\frac{1+d}{2}\Bigr) = 1 - d^2 /(2\log 2) + O(d^4).
\end{equation}
So, we can show the following result of the channel capacity for an extremal BSC:

\begin{corollary}[extremal BSC]\label{cor1}
Given a BSC channel with crossover probability
 $p = (1 + d)/2$, if $d$ is small (i.e., $|d|\ll 1$), then, 
$C\approx c_0\cdot d^2$, where the constant $c_0=1/(2\log 2)$.
\end{corollary}

Therefore, 
we can send one bit with an arbitrarily low probability of error with the minimum number of transmissions
 $1/C= (2\log 2)/{d^2}$, i.e., $O(1/d^2)$.
Interestingly, in communication theory, 
this extremal BSC is rare as 
we typically deal with $|d| \gg 0$ (see \cite{Shokrollahi}).

\section{Sampling Theorems with Incomplete Signals}
\label{sect_main_result}

In this section, 
we 
put forward two sampling problems, namely, the classical and generic versions.
Without loss of generality,
we assume that the discrete statistical signals are not restricted to a particular application domain
and the signals are $2^n$-valued for a fixed $n$.

Specifically,
we give the mathematical model on the signal represented by
an arbitrary (and not necessarily deterministic) function $F$
as follows.
Let $X$ be the $n$-bit output sample of $F$, 
assuming that the input is random and uniformly distributed.
Denote the output distribution of $X$ by $f$.
Note that our assumption on a general setting of
 discrete statistical signals
is described by the assumption that $F$ is an arbitrary yet fixed function with the $n$-bit output.

The classical sampling problem can be formally stated as follows.

\begin{theorem}[Classical Sampling Problem]\label{thm_5}
Assume that the largest Walsh coefficient of $f$ is $d=\widehat{f}(m_0)$ for a 
nonzero $n$-bit vector $m_0$.
We can detect signals represented by $F$ with an arbitrarily low probability of error, 
using minimum number $N=(8\log 2)/{d^2}$ of samples of $F$,
i.e., $O(1/{d^2})$.
\end{theorem}

Note that it can be interpreted as 
the classical distinguisher which is used often in statistical cryptanalysis, though the problem
statement of the classical distinguisher is slightly different and it 
 uses a slightly different $N$ 
 (cf. \cite{vaudenay_textbook2006}).
The classical sampling problem
 assumes that $F$ together with its characteristics (i.e., the largest Walsh coefficient $d$) are known \emph{a priori}.

Next, we present our main sampling results for practical (and widely applicable) sampling.
Assuming that it is infeasible to know signal represented by $F$ \emph{a priori}, 
we want to use the bounded number of samples to detect signals with an arbitrarily low probability of error.
Note that the sampled signal is often incomplete\footnote{that is, it is possible that not all the outputs are generated by sampling.} and so the associated distribution is not precise.
 We call this problem as generic sampling with incomplete noisy signals.
In analogy to the classical distinguisher, this result can be interpreted as a generalized distinguisher\footnote{With $n=1$,
this appears as an informal
result in symmetric cryptanalysis, which is used as a black-box analysis tool in several crypto-systems.}  
 in the context of statistical cryptanalysis.
We give our main result with $n=1$ below.

\begin{theorem}[Generic Sampling Problem with $n=1$]\label{cor_thm_5}
Assume that the sampling size of $F$ is upper-bounded by $N$.
Regardless of the input size of $F$,
in order to detect the signal $F$ with an arbitrarily low probability of error, 
 it is necessary and sufficient to have the following condition satisfied,
i.e.,
 $f$ has a nontrivial Walsh coefficient $d$ with 
$|d|\ge c/{\sqrt{N}}$, where the constant $c=\sqrt{8\log 2}$.
\end{theorem}

Assume that $f$ satisfy the following conditions:
1) the cardinality of the support of $f$ is a power of two (i.e., $2^n$), and
2) $2^n$ is small, and
3) $f(i)\in (0, 3/2^n)$, for all $i$.
Now, we present a generalized result for $n\ge 1$, 
which incorporates Theorem \ref{cor_thm_5} as a special case:

\begin{proposition}[Generic Sampling Problem with $n\ge 1$]\label{cor_thm_5b}
Assume that the sampling size of $F$ is upper-bounded by $N$.
Regardless of the input size of $F$,
in order to detect the signal $F$ with an arbitrarily low probability of error, 
 it is necessary and sufficient to have the following condition satisfied,
i.e., 
\begin{equation}\label{E_condition_generic_sampling}  
\sum_{i\ne 0} (\widehat{f}(i))^2 \ge (8\log 2)/N.
\end{equation}
\end{proposition}

We note that the sufficient condition can be also proved based on results of the classic distinguisher (i.e., Squared Euclidean Imbalance),
 which uses the notion of Kullback-Leibler distance and states that $\sum_{i\ne 0} (\widehat{f}(i))^2 \ge (4\log 2)/N$ is required for a high probability of success \cite{vaudenay_textbook2006}.
Secondly, from (\ref{E_condition_generic_sampling}),
 the discrete statistical signals can be characterized 
by large Walsh coefficients of the associated distribution.
Thus the most significant transform-domain signals are the largest coefficients in our generalized model.

\subsection{Proof of Theorem \ref{cor_thm_5}}

We first prove the following hypothesis testing result by Shannon's Channel Coding Theorem:

\begin{theorem}\label{thm_3}
Assume that the boolean random variable $\mathcal{A}$ has bias $d$ and $d$ is small. We are
given a sequence of random samples, 
which are 
i.i.d. following the distribution of either $\mathcal{A}$ or a uniform distribution. 
We can tell the sample source with an arbitrarily low probability of error,
 using the minimum number $N$ of samples $(8\log 2)/{d^2}$, i.e.,
$O(1/d^2)$.
\end{theorem}

\begin{proof}
We propose a novel non-symmetric binary channel. %
Assume the channel with the following transition matrix
\begin{equation}
p(y|x)=\left(
\begin{array}{cc}
1-p_e & p_e \\
1/2 & 1/2
\end{array}
\right),
\end{equation}
where $p_e=(1-d)/2$ and $d$ is small.
The matrix entry in the $x$th row and the $y$th column denotes the conditional
probability that $y$ is received when $x$ is sent.
So, the input bit $0$ is transmitted by this channel with error probability $p_e$
(i.e., the received sequence has bias $d$ if input symbols are 0)
 and the input bit $1$ is transmitted with error probability 1/2
(i.e., the received sequence has bias $0$ if input symbols are 1).
By Shannon's channel coding theorem, with a minimum number
of $N = 1/C$ transmissions, we can reliably (i.e., with an
arbitrarily low probability of error) detect the signal source (i.e., determine whether the
input is `0' or `1').

To compute the channel capacity $C$, i.e., find the maximum
 defined in (\ref{E_capacity_def}),
 no closed-form solution exists in general. Nonlinear optimization algorithms 
(see \cite{Arimoto1972channel_capacity,Blahut1972channel_capacity})
 are known to find a numerical solution. Below, 
we propose a simple method to give a closed-form estimate $C$ for our extremal binary channel.
As $I(X;Y)= H(Y)-H(Y|X)$, we first compute $H(Y)$ by
\begin{equation}\label{E_tmp1}
H(Y) = H\Bigl( p_0(1-p_e)+ (1-p_0)\times \frac{1}{2} \Bigr),
\end{equation}
where $p_0$ denotes $p(x=0)$ for short.
Next, we compute %
\begin{eqnarray}
H(Y|X) &=& \sum_x p(x)H(Y|X=x) \nonumber\\
&=& p_0\Bigl( H(p_e)-1 \Bigr) +1.\label{E_tmp2}
\end{eqnarray}
Combining (\ref{E_tmp1}) and (\ref{E_tmp2}), we have
\begin{equation}
I(X;Y) = H\Bigl( p_0\times \frac{1}{2} - p_0p_e + \frac{1}{2} \Bigr) - 
  p_0H(p_e) + p_0 -1.
\end{equation}
As $p_e=(1-d)/2$, we have
\begin{equation}
I(X;Y)
=H(\frac{1 + p_0d}{2}) -p_0\Bigl( H(\frac{1-d}{2})-1 \Bigr) - 1.
\end{equation}
For small $d$,
 we apply (\ref{approx_H_special}) in Appendix
\begin{equation}\label{E_tmp3}
I(X;Y) = -\, \frac{p_0^2 d^2}{2\log 2} - p_0\Bigl( H(\frac{1-d}{2})-1 \Bigr) + O(p_0^4 d^4).
\end{equation}
 Note that
 the last term $O(p_0^4 d^4)$ on the right side of (\ref{E_tmp3}) is ignorable.
 Thus, $I(X;Y)$ is estimated to approach the maximum when
\begin{equation}
p_0 = -\,\frac{ H(\frac{1-d}{2})-1 }{d^2/(\log 2)} \approx \frac{d^2/(2\log 2)}{d^2/(\log 2)}=\frac{1}{2}.
\end{equation}
Consequently, we estimate the channel capacity (\ref{E_tmp3}) by
\begin{eqnarray*}
C &\approx& -\,\frac{1}{4}d^2/(2\log 2) +\frac{1}{2} \Bigl(1-H(\frac{1-d}{2})\Bigr)\\ 
 &\approx& -\, d^2/(8\log 2) + d^2/(4\log 2),
\end{eqnarray*}
which is $d^2/(8\log 2)$.
\end{proof}

We now proceed to prove Theorem \ref{cor_thm_5}.
The only nontrivial Walsh coefficient $d$ for $n=1$ is $\widehat{f}(1)$, 
which is the bias of $F$.
First, we will show by contradiction that this is a necessary condition.
That is, if we can identify $F$ with an arbitrarily low probability of error, 
then, we must have $|d|\ge c/{\sqrt{N}}$.
Suppose $|d| < c/{\sqrt{N}}$ otherwise.
Following the proof of Theorem \ref{thm_3}, we know that the error probability is bounded away from zero 
as the consequence of Shannon's Channel Coding Theorem. This is contradictory.
Thus, we have shown that the condition on $d$ is a necessary condition.
Next, we will show that it is also a sufficient condition.
That is, if $|d|\ge c/{\sqrt{N}}$, then, 
we can identify $F$ with an arbitrarily low probability of error.
This follows directly from Theorem \ref{thm_5} with $n=1$.
We complete our proof. 
\subsection{Proof of Proposition \ref{cor_thm_5b}}

Assume that the channel have transition matrix $p(y|x)$.
Let $p(y|x=0)$ denote the distribution $f$, and
let $p(y|x=1)$ be a uniform distribution.
Denote the channel capacity by $C$.
For convenience, we let $f(i)=u_i + 1/{2^n}$
for $i=0,\ldots,2^n-1$.
 Note that we have $\sum_i u_i =0$ and $- 1/{2^n} < u_i < (2^n-1)/{2^n}$ for all $i$.

By Taylor series, for all $i$ we have
\begin{eqnarray}
\log(u_i + \frac{1}{2^n})
&=& \log \frac{1}{2^n} + 2\Bigl(\frac{u_i}{\frac{2}{2^n} + u_i} + \frac{1}{3}(\frac{u_i}{\frac{2}{2^n} + u_i} )^3\nonumber \\ && +\cdots \Bigr)\nonumber \\
 &\approx& \log \frac{1}{2^n} + \frac{2u_i}{\frac{2}{2^n} + u_i}, \label{Eq_a1}
\end{eqnarray}
as we know $u_i/(\frac{2}{2^n} + u_i) \in (-1,1)$.
And we deduce that with small $2^n$, we can calculate $-\log 2 \cdot H(f)$ by
\begin{eqnarray}
&& \sum_i (u_i + \frac{1}{2^n}) \log(u_i + \frac{1}{2^n})\nonumber \\ &\approx&
  \log \frac{1}{2^n} + 2\sum_i u_i - \sum_i \frac{2u_i}{2 + 2^n\cdot u_i}\nonumber \\
&\approx& 
\log \frac{1}{2^n} -
\frac{1}{2^{n-1}} \sum_i (1 - \frac{1}{1 + 2^{n-1}u_i}).
\label{Eq_a2}
\end{eqnarray}
Assuming that $|2^{n-1}u_i| < 1$ for all $i$ (and small $2^n$),
we have 
\begin{equation}
\sum_i \frac{1}{1 + 2^{n-1}u_i} \approx \sum_i \Bigl(1 - 2^{n-1}u_i + (2^{n-1}u_i)^2 \Bigr).
\end{equation}
We continue (\ref{Eq_a2}) by
\(
-\log 2 \cdot H(f) \approx \log \frac{1}{2^n} + 2^{n-1} \sum_i u_i^2 .
\)
Meanwhile, by property of Walsh transform (cf. \cite[Sect.2]{my_new_submission}), we know
\begin{equation}
 \sum_i (\widehat{f}(i))^2 = 2^n \sum_i \Bigl( f(i) \Bigr)^2 
=  2^n \sum_i u_i^2 + 1 .
\end{equation}
So, we have shown an important result as follows
\begin{eqnarray}
H(f) &\approx& n - \frac{2^n}{2\log 2} \sum_i \Bigl(f(i)-\frac{1}{2^n}\Bigr)^2\nonumber \\
 &=& n - \frac{\sum_{i\ne 0} (\widehat{f}(i))^2}{2\log 2}, \label{Eq_H_f}
\end{eqnarray}
assuming that
1) the cardinality of the support of $f$ is a power of two (i.e., $2^n$), and
2) $2^n$ is small, and
3) $f(i)\in (0, 3/2^n)$, for all $i$.

Next, in order to calculate $C$, by (\ref{Eq_H_f}) we first compute
\begin{equation}
\label{Eq_H_Y_given_X}
H(Y|X) = p_0H(f) + (1-p_0)n \approx n - \frac{p_0\sum_{i\ne 0} (\widehat{f}(i))^2}{2\log 2},
\end{equation}
where $p_0$ denote $p(x = 0)$ for short.
Denote the distribution of $Y$ by $D_Y$. We have
$D_Y(i) = p_0f(i) + (1-p_0)/2^n = p_0u_i + 1/2^n$ for all $i$. Again we can apply (\ref{Eq_H_f}) and get
\begin{equation}
\label{Eq_H_Y}
H(Y) \approx n - \frac{\sum_{i\ne 0} (\widehat{D_Y}(i))^2}{2\log 2} = 
n - \frac{p_0^2\sum_{i\ne 0} (\widehat{f}(i))^2}{2\log 2} .
\end{equation}
So, we have
\begin{equation}
I(X;Y) = H(Y) - H(Y|X) \approx \frac{(p_0 - p_0^2)\sum_{i\ne 0} (\widehat{f}(i))^2}{2\log 2} .
\end{equation}
When $p_0=1/2$, we have the maximum $I(X;Y)$, which equals
\begin{equation}
\label{Eq_new_C}
C \approx \frac{\sum_{i\ne 0} \Bigl(\widehat{f}(i)\Bigr)^2}{8\log 2} = \frac{2^n \sum_i \Bigl(f(i)-\frac{1}{2^n}\Bigr)^2}{8\log 2}.
\end{equation}
Consequently, we have
$N\ge 1/C$, i.e.,
$\sum_{i\ne 0} (\widehat{f}(i))^2 \ge (8\log 2)/N$.
And this is a necessary and sufficient condition, following Shannon's theorem.

\subsection{More Generalized Results}

Above we consider the case that $f(i)$ is not so far from the uniform distribution.
Based on the weaker assumption that $2^n$ is small and $f(i)= 1/2^n + u_i >0$ (i.e., $u_i > - 1/2^n$) for all $i$, 
 we now show a a more general result.
We have
\begin{equation}
H(f) \approx n + \frac{1}{\log 2} \sum_i \frac{u_i}{1 + 2^{n-1}u_i}
\end{equation}
by (\ref{Eq_a2}).
Following similar computations we have
\begin{eqnarray}
H(Y|X) &\approx& n + \frac{p_0}{\log 2} \sum_i \frac{u_i}{1 + 2^{n-1}u_i}, \\
H(Y) &\approx& n + \frac{1}{\log 2} \sum_i \frac{p_0u_i}{1 + 2^{n-1}p_0 u_i} .
\end{eqnarray}
So, we obtain the following general result,
\begin{equation}
\label{Eq_new_general_C}
C \approx \max_{p_0} \frac{1}{\log 2} \sum_i \Bigl(\frac{p_0 u_i}{1 + 2^{n-1}p_0 u_i}\, - \,\frac{p_0 u_i}{1 + 2^{n-1}u_i}\Bigr).
\end{equation}
Recall that if $|u_(i)| < 2/2^n$ for all $i$, (\ref{Eq_new_general_C}) can be approximated by (\ref{Eq_new_C}), 
which is achieved with $p_0=1/2$. Specifically, if $|u_(i)| < 2/2^n$,
the approximation for the addend in (\ref{Eq_new_general_C})
can be expressed as follows,
\begin{equation}
\label{Eq_a3}
(\frac{p_0 u_i}{1 + 2^{n-1}p_0 u_i}\, - \,\frac{p_0 u_i}{1 + 2^{n-1}u_i} ) \approx
 2^{n-1}u_i^2 (p_0 - p_0^2),
\end{equation}
where we use 
$\frac{v}{1+v}=1-\frac{1}{1+v}\approx 1-(1 - v + v^2)=v-v^2$ (for $|v|<1$). 

Note that for $|v|>1$,
we have $\frac{v}{1+v} \approx 1 - 1/v + 1/v^2$.
We can show that with $2^{n-1} u_i = k > 1$ for some $i$, the addend in (\ref{Eq_new_general_C}) can achieve the maximum 
when $p_0=1/k$, that is, 
\begin{equation}
\label{Eq_a3b}
\max_{p_0}
(\frac{p_0 u_i}{1 + 2^{n-1}p_0 u_i}\, - \,\frac{p_0 u_i}{1 + 2^{n-1}u_i} ) \approx
\frac{1}{2^{n-1}}(1-\frac{1}{k} + \frac{1}{k^2}-\frac{1}{k^3}) .
\end{equation}
On the other hand, the right-hand side of (\ref{Eq_a3}) equals
$(p_0 - p_0^2)k^2 /2^{n-1}$,
which is much larger than (\ref{Eq_a3b}).

Meanwhile, with $2^{n-1} u_i = k = 1$ for some $i$,
we have
\begin{eqnarray}
&& \max_{p_0}
(\frac{p_0 u_i}{1 + 2^{n-1}p_0 u_i}\, - \,\frac{p_0 u_i}{1 + 2^{n-1}u_i} ) \nonumber\\ 
&\approx&
\max_{p_0} \frac{1}{2^{n-1}}(\frac{p_0}{2}- p_0^2) 
 = \frac{1}{16} \cdot\frac{1}{2^{n-1}}, \label{Eq_a3c}
\end{eqnarray}
when $p_0=1/4$.

\subsection{Further Discussions}
Based on Renyi's information measures (cf. \cite{renyi}),
we make a conjecture\footnote{see \cite{FPS2017capacity} for most recent results on the conjecture.} below for an even more general form of our channel capacity $C$.
Recall that Renyi's information divergence (see \cite{renyi}) of order $\alpha = 1/2$ 
of distribution $P$ from another distribution $Q$ on a finite set $\mathcal{X}$ is defined as
\begin{equation}
D_{\alpha} (P \| Q) \stackrel{\text{def}}{=} 
  \frac{1}{\alpha - 1} \log \sum_{x \in \mathcal{X}} P^{\alpha}\bigl( x \bigr) Q^{1-\alpha}\bigl( x \bigr) .  
\end{equation}
So, with $\alpha = 1/2$,
\begin{equation}
D_{1/2} (P \| Q) \stackrel{\text{def}}{=}  (-2) \log \sum_{x \in \mathcal{X}} \sqrt{P(x)Q(x)} .
\end{equation}

\begin{conjecture}
Let $Q,U$ be a non-uniform distribution and a uniform distribution over the support of cardinality $2^n$.
Let the matrix of $T$ consist of two rows $Q,U$ and $2^n$ columns.
We have the following relation between Renyi's divergence of degree $1/2$ and the generalized channel capacity of degree $1/2$
 (i.e., standard Shannon's channel capacity),
\begin{equation}
D_{1/2}(Q\|U)=2\cdot C_{1/2}(T).
\end{equation}
\end{conjecture}

Finally, 
assume that $2^n$-valued $F$ has potentially large input space $2^L$.
The collected normalized distribution of $N \gg 2^n$ output samples 
in the time-domain approximately fits in the noisy model of
 \cite{isit2014wht,isit2015wht}. 
That is, the additive Gaussian noise is i.i.d. and has zero mean and variance $\sigma^2=1/(N2^n)$.
In this case, 
noisy sparse WHT \cite{isit2015wht} can be used for our generic sampling problem
in recovering the nonzero Walsh coefficients $\widehat{f}(i)$ and index positions $i$. 
And
we obtain evaluation time $N$ queries of $F$, processing time roughly on the order of $n$,
provided that 
\begin{equation}
\sum_{i\ne 0} (\widehat{f}(i))^2 \ge (8\log 2)/N,
\end{equation}
that is, \\
\begin{equation}
\text{SNR}
= \sum_{i \ne 0} \Bigl(\widehat{f}(i)\Bigr)^2/(2^n \cdot \sigma^2) 
= \sum_{i \ne 0} \Bigl(\widehat{f}(i)\Bigr)^2/(\frac{1}{N}) 
\ge 8\log 2.
\end{equation}

\section{Applications and Experimental Results}
\label{sect_experiments}

We first demonstrate a cryptographic sampling application.  
The famous block cipher GOST 28147-89 is a balanced
Feistel network of 32 rounds. 
It has a 64-bit block size and a key length of 256 bits.
Let 32-bit $L_i$ and $R_i$ denote the left and
right half at Round $i\ge 1$ and $L_0,R_0$ denote the plaintext.
The subkey for Round $i$ is denoted by $k_i$. 
For the purpose of multi-round analysis, 
our target function is $F(R_{i-1}, k_i) = R_{i-1}\oplus k_i\oplus f(R_{i-1}, k_i)$, where $f(R_{i-1}, k_i)$ is the round function.
 We choose $N=2^{40}$.
It turns out that the largest three Walsh coefficients are $2^{-6},2^{-6.2},2^{-6.3}$ respectively.
This new weakness leads to various severe attacks \cite{FPS2017gost} on GOST.
Similarly, our cryptographic sampling technique is applicable and it
 further threatens the security of the SNOW 2.0 cipher \cite{zhangbin_crypto2015}.

For general $\text{SNR} \ge 8\log 2$,
regardless of $2^L$, when $N \sim 2^n$,
we choose appropriate $b$ such that $N=b\cdot 2^\ell$. 
So, roughly $b\cdot \ell$ processing time 
 is needed. 
With $L=64$, it would become a powerful universal analytical tool\footnote{Note that
current computing technology \cite{acm_exascale} can afford exascale WHT (i.e., on the order of $2^{60}$) within 2 years,
which uses $2^{15}$ modern PCs.} 
 for 
security evaluation of
the core building blocks of symmetric cryptography (cf. \cite{my_new_submission,vaudenay_new}).

Another notable practical application is software performance optimization.  
Modern large-scale IT systems are of hybrid nature.
Usually, only partial information about the architecture as well as some of its component units is known.
Further, the revolutionary change of the physical components (e.g., main memory, storage) inevitably
demand that the system take full advantages of the new hardware characteristics.
We expect that our sampling techniques 
(and the transform-domain analysis of the running time)
would help with performance analysis and optimization for the whole heterogeneous system.

\section{Conclusion}
\label{sect_end}
In this paper,
we model general discrete statistical signals as the output samples of
an unknown arbitrary yet fixed function (which is the signal source).
We translate Shannon's channel coding theorem
to solve a hypothesis testing problem.
The translated result allows to solve a generic sampling problem, 
for which we know nothing about the signal source \emph{a priori} 
and we can only afford bounded sampling measurements.
Our main results demonstrate that the classical signal processing tool of Walsh transform is essential: 
it is the large Walsh coefficient(s) that characterize(s) discrete statistical signals, 
regardless of the signal sources. 
By Shannon's theorem, 
we establish the \emph{necessary and sufficient} condition 
for the generic sampling problem under the general assumption of statistical signal sources.
As for future direction of this work,
it is interesting to 
investigate noisy sparse WHT for our generic sampling problem in a more general setting $N \ll 2^n$ and $2^n$ is large,
in which case the noise cannot be modelled as Gaussian.

\end{document}